\documentclass[ 11pt,
 onecolumn,  triplespace,
journal]{IEEEtran}

\usepackage{cite}
\usepackage{amsmath,amssymb,amsfonts}
\usepackage{graphicx}
\usepackage{textcomp}
\usepackage{xcolor}
\def\BibTeX{{\rm B\kern-.05em{\sc i\kern-.025em b}\kern-.08em
    T\kern-.1667em\lower.7ex\hbox{E}\kern-.125emX}}
\usepackage{amssymb}
\usepackage{amsthm} 
\usepackage{color}
\usepackage{cite}
\usepackage{enumerate}
\usepackage{fancyhdr}
\usepackage{graphicx}
\usepackage{bm}
\usepackage{soul,xcolor}
\def\x{{\mathbf x}}

\newtheorem{theorem}{Theorem}[section]

\newtheorem{definition}[theorem]{Definition}
\newtheorem{corollary}[theorem]{Corollary}

\usepackage{algorithm}
\usepackage{algorithmic}

\numberwithin{equation}{section}

\title{Kalman Filtering of Stationary Graph Signals
}

\author{Yang Chen, Yeonju Lee, Yao Shi and Qiyu Sun
\thanks{Chen is with School of Mathematics and Statistics and LCSM, Hunan Normal University,  Changsha, Hunan 410081, P. R. China; 
 Lee is with Division of Applied Mathematical Sciences
Korea University, Sejong Campus,
Sejong City  30072, South Korea; and
Shi and Sun are with Department of Mathematics, University of Central Florida, Orlando, Florida 32816, USA;
Emails: ychenmath@hunnu.edu.cn; leeyeonju08@korea.ac.kr; Yao.Shi@ucf.edu;   qiyu.sun@ucf.edu.\\
This project is partially supported by the National Natural
Science Foundation of China (12171490), the Research Foundation of Education Bureau of Hunan Province, China (24B0107) and the Natural
Science Foundation of Hunan Province, China (2025JJ50008) and   the National Research Foundation of Korea (NRF) grant 
funded by the Korea government (MSIT) 
(2021R1A2C1008360).}}

\begin{document}

\maketitle

\begin{abstract}

In this paper, we propose a novel definition of stationary graph signals, formulated with respect to a symmetric graph shift, such as the graph Laplacian. 
We show that stationary graph signals can be generated by transmitting white noise through polynomial graph  channels, and that their stationarity is 
preserved under  polynomial channel transmission.

In this paper, we also investigate Kalman filtering to dynamical systems characterized by polynomial state and observation matrices. We demonstrate that Kalman filtering maintains the stationarity of graph signals, while effectively incorporating both system dynamics and noise structure. In comparison to the static inverse filtering method and naive zero-signal strategy, the Kalman filtering procedure yields more accurate and adaptive signal estimates, highlighting its robustness and versatility in graph signal processing.

\end{abstract}


\section{Introduction}

Graph Signal Processing (GSP) provides a  versatile, powerful and 
mathematically grounded framework for representing, analyzing, and manipulating datasets that reside on networks and 
 irregular domains \cite{Shuman2013, Ortega2018, 
Dong2020, 
Isufi2024}. In this paper, we introduce a novel definition of stationary graph signals and consider Kalman filtering
to dynamical systems characterized by state and observation matrices being polynomial filters of a graph shift.  
 
Stationarity is a cornerstone of many signal analysis techniques, serving as a critical assumption that enables tractable modeling, spectral estimation, and filter design. In classical signal processing, stationary signals are characterized by covariance structures that remain invariant under time shifts. Extending this concept to  graph domain, the notion of stationary graph signals has been introduced  \cite{Girault2015, girault2017,  marques2017,  Perraudin2017, segarra2018, Jian2022, Zheng2025}.
 Within this framework, the covariance of random signals is invariant under a graph shift, such as the graph Laplacian. However, unlike in the temporal setting, existing definitions of stationary graph signals do not guarantee that such signals can be generated by passing  white noises through some polynomial graph filters, particularly when the graph shift has repeated eigenvalues. This limitation motivates us to introduce a new definition of stationarity, designed to guarantee that every stationary graph signal admits a generative model constructed from white noise and polynomial graph filtering; see Theorem \ref{gaussian.thm}.

The classical Kalman filter is a foundational algorithm for estimating the hidden states of dynamical systems from noisy observations, with widespread applications in aerospace, robotics, and signal processing \cite{Kalman1960, chui2009, humpherys2012}. In recent years, graph Kalman filtering has emerged as a powerful extension, enabling the tracking of dynamical signals over networks \cite{shi2009, sagi2023, buchnik2024, buchnik2024b}.
In this paper, we study graph Kalman filtering where both the system dynamics and observation models are expressed as polynomials  of a graph shift (and hence the state signals are stationary). We show that this approach yields more accurate and adaptive estimates than static inverse filtering methods and naive zero-signal strategy; see Section \ref{kalman.section}.

\section{Stationary graph signals}
\label{sec:preliminaries}

Graphs provide a flexible tool to represent the topology of networks and  irregular domains. In this paper, we assume that the graph ${\mathcal G}$ to model the underlying topological structure  is undirected and finite. The graph  ${\mathcal G}$  comprises a vertex set $V$, an edge set $E$, and a weighted adjacency matrix ${\bf W}=[W(i, j)]_{i,j\in V}$. 

Graph shifts are the backbone of  GSP, serving as the graph-domain analogue of time shifts in classical signal processing. Here a {\em graph shift} ${\bf S}$ on the graph ${\mathcal G}$, represented by a matrix ${\bf S}=[S(i,j)]_{i,j\in V}$, has its entries taking zero value except possibly on the diagonal $i=j$ or at endpoint pair $(i,j)$ of an edge in  $E$.
The illustrative examples are  the weighted adjacency matrix ${\bf W}$, the  degree matrix ${\bf D}$,  the graph Laplacian $\mathbf{L} = \mathbf{D} - \mathbf{W}$
 and their variants, where  the diagonal  degree matrix ${\bf D}$ has diagonal entries  $D(i, i)=\sum_{j\in V} W(i,j), i\in V$ \cite{Ortega2022, Emirov2022}. In this paper, we fix a graph shift ${\mathbf S}$ that is symmetric and real-valued.

Graph filters are fundamental tools in GSP designed to amplify or suppress components of graph signals, much like classical filters operating on time or spatial signals.  We say that a  graph filter ${\bf H}$ is a  {\em polynomial  filter} of the graph shift ${\bf S}$ if
\begin{equation}\label{polyfilter.def}
    \mathbf{H}=h(\mathbf{S})=h_0 \mathbf{I} + h_1 \mathbf{S} + \dots + h_{L}\mathbf{S}^{L}
\end{equation}
for some polynomial $h(t) = h_0 + h_1 t + \dots + h_{L}t^{L}$, where ${\bf I}$ is the identity matrix of appropriate size.
A key advantage of polynomial graph filters lies in their efficient implementation within the spatial domain, which involves only a finite number of local operations, with data exchanged between neighboring vertices. This locality makes polynomial filters 
well-suited for distributed networks and scalable architectures \cite{Emirov2022, Sandryhaila2013}.

For the graph shift ${\bf S}$ on the graph ${\mathcal G}$, we have
\begin{equation}\label{eigen.dec}
{\bf S}={\bf U} {\pmb\Lambda} {\bf U}^\top =\sum_{n=1}^N \lambda(n) {\bf u}_n {\bf u}_n^\top,
\end{equation}
where $N$ is the order of the  graph ${\mathcal G}$,  ${\bf U}=[{\bf u}_1, \ldots, {\bf u}_N]$
is an orthogonal matrix, and the diagonal matrix ${\pmb\Lambda}={\rm diag}[\lambda(1),\ldots, \lambda(N)]$ has eigenvalues of the graph shift ${\bf S}$ as its diagonal entries.
For the polynomial filter  ${\bf H}$ in \eqref{polyfilter.def}, applying the eigendecomposition 
\eqref{eigen.dec}  repeatedly yields
\begin{equation} \label{filter.Fourier}
{\bf H}= {\bf U} h({\pmb \Lambda}) {\bf U}^\top=\sum_{n=1}^N
h(\lambda(n)) {\bf u}_n {\bf u}_n^\top,
\end{equation}
where $h(\lambda(n)), 1\le n\le N$, are known as the frequency responses of the polynomial filter  ${\bf H}$.

Stationary graph signal is a fascinating extension of classical stationarity concept to  graph signals.
In this paper, we introduce a novel definition of stationary graph signals.

\begin{definition}\label{stationary.def}
A random graph signal $\x$ is  stationary if
it has zero mean ${\mathbb E}{\mathbf x}={\bf 0}$ and
its covariance matrix ${\rm cov}(\x)={\mathbb E}\big((\x-{\mathbb E}\x) (\x-{\mathbb E}\x)^\top)=
\mathbb E(\x\x^\top)
$ is a polynomial filter  of  the graph  shift  ${\bf S}$.
\end{definition}

Similar to  stationary signals in the classical signal processing with their covariance invariant under time shifts, we see that the covariance matrix of a stationary graph signal is shift-invariant with respect to the graph shift ${\bf S}$, i.e., 
${ \mathbf S}  {\rm cov}({\mathbf x}) = {\rm cov}({\mathbf x}) \mathbf{S}
$. 
We remark that  a shift-invariant covariance matrix need not, in general, correspond to a polynomial filter. For instance, on the unweighted cycle graph, it can be verified that a graph filter 
is a polynomial filter of the graph Laplacian if and only if it can be represented by a symmetric circulant matrix, while it is shift-invariant with respect to the graph Laplacian if and only if it is a  circulant matrix, regardless of symmetry. 
On the other hand, a shift-invariant filter is a polynomial filter provided that all eigenvalues $\lambda(n), 1\le n\le N$, of the graph shift ${\bf S}$  are distinct \cite{Emirov2022}.
In the literature, stationary graph signals are typically defined to have zero mean and covariance matrix being shift-invariant with respect to the graph  shift  ${\bf S}$
\cite{Girault2015,  Perraudin2017, segarra2018,  Jian2022,  Zheng2025, girault2017}.

We say that ${\bf e}$  is a standard white noise on the graph 
${\mathcal G}$ if its mean and covariance satisfy $\mathbb E({\bf e})={\bf 0}$  and ${\rm cov}({\bf e})={\bf I}$. In the following theorem, we  show that any stationary graph signal can be modeled as the output of a polynomial graph filter applied to  a white  noise.

\begin{theorem} \label{gaussian.thm}
A random graph signal ${\bf x}$ on the graph ${\mathcal G}$  is stationary if and only if
\begin{equation}\label{gaussian.thm.eq1} \x={\bf H}{\bf e}\end{equation}
where ${\bf H}$ is a  polynomial filter and
${\bf e}$ is a standard white noise on the graph ${\mathcal G}$.
\end{theorem}

\begin{proof}
For a random signal ${\bf x}$ with the expression \eqref{gaussian.thm.eq1}, we have
$ {\mathbb E}({\bf x})={\bf H} {\mathbb E}({\bf e})={\bf 0}
$
and
${\mathbb E}({\bf x}{\bf x}^\top)={\bf H} 
{\mathbb E}({\bf e} {\bf e}^\top) {\bf H}^\top= {\bf H} {\bf H}^\top={\bf H}^2
$.  
This proves that ${\bf x}$ is stationary and hence
the sufficiency.

Now we establish the necessity. Let $\x$ be a stationary graph signal with ${\rm cov}(\x)=h({\mathbf S})$ for some polynomial filter in \eqref{polyfilter.def}. By \eqref{filter.Fourier} and the positive semi-definiteness of the covariance matrix ${\rm cov}(\x)$, we see that the frequency responses of the  polynomial filter ${\rm cov}(\x)$ are non-negative, i.e.,
 \begin{equation} \mu(n):= h(\lambda(n)) \ge 0, \ 1\le n\le N. \end{equation}
Applying the Lagrange interpolation at the set of all distinct
eigenvalues of the graph shift ${\bf S}$, we can find a polynomial
$g$ that matches the square root $\sqrt{\mu(n)}$ of the 
frequency
responses of the polynomial filter ${\bf H}$
 at each eigenvalue,
\begin{equation}\label {gaussian.thm.pfeq1}
g(\lambda(n))=\sqrt{\mu (n)}, \ \ 1\le n\le N.
\end{equation}

Define
    \begin{equation}\label {gaussian.thm.pfeq2}
        {\bf e}:=\sum_{\mu(n)\ne 0} 
        \frac{1}{\sqrt{\mu(n)}}
     {\bf u}_n {\bf u}_n^T {\bf x} +
     \sum_{\mu(n)= 0} \tilde {x}_n  {\bf u}_n, 
    \end{equation} 
    where $\tilde x_n\sim \mathcal  N(0, 1)$, 
indexed by $n$ satisfying $\mu(n)=0$, are i.i.d. Gaussian noises on the real line and independent of the random signal ${\bf x}$. Then 
$$\mathbb E({\bf e})=
\sum_{\mu(n)\ne 0} 
        \frac{1}{\sqrt{\mu(n)}}
     {\bf u}_n {\bf u}_n^T \mathbb E({\bf x}) +
     \sum_{\mu(n)= 0} {\mathbb E}(\tilde {x}_n)  {\bf u}_n= 
     {\bf 0},$$ 
and 
    \begin{eqnarray*}
    {\rm cov} (\bf e) & \hskip-0.08in  = & \hskip-0.08in   
    \sum_{\mu(n)\ne 0, \mu(n')\ne 0}
    \frac{ {\bf u}_n {\bf u}_{n'}^\top }{\sqrt{\mu(n)\mu(n')}}
     {\mathbb E}\big( {\bf u}_n^\top {\bf x}{\bf x}^\top {\bf u}_{n'}\big)  +  \sum_{\mu(n)\ne 0, \mu(n')=0}
    \frac{ {\bf u}_n {\bf u}_{n'}^\top }{\sqrt{\mu(n)}}
     {\mathbb E}\big( \tilde { x}_{n'} {\bf u}_n^\top {\bf x}  \big) \\  
       & & +       \sum_{\mu(n)=0, \mu(n')\ne 0}
    \frac{ {\bf u}_n {\bf u}_{n'}^\top }{\sqrt{\mu(n')}}
     {\mathbb E}\big( \tilde {x}_{n}  {\bf u}_{n'}^\top {\bf x} \big)  +       \sum_{\mu(n)=\mu(n')=0}
 {\bf u}_n {\bf u}_{n'}^\top 
     {\mathbb E}\big( \tilde 
     {x}_n \tilde {x}_{n'} \big) \\  
       & \hskip-0.08in  = &  \hskip-0.08in     \sum_{\mu(n)\ne 0}
     {\bf u}_n {\bf u}_n^\top+ {\bf 0}+{\bf 0}+
      \sum_{\mu(n)=0}
     {\bf u}_n {\bf u}_n^\top={\bf I}.
\end{eqnarray*}
Hence the random graph signal ${\bf e}$ in \eqref{gaussian.thm.pfeq2} is a white  noise.

Define the polynomial filter ${\bf G}$ by
\begin{equation}
{\bf G}=g({\bf S})=\sum_{n=1}^N g(\lambda(n)) {\bf u}_n {\bf u}_n^T,
\end{equation}
where the polynomial $g$ is given in \eqref{gaussian.thm.pfeq1}. Then 
\begin{eqnarray*}
{\bf G} {\bf e} & \hskip-0.08in = & \hskip-0.08in 
\sum_{\mu(n)\ne 0}
\frac{g(\lambda(n))}{\sqrt{\mu(n)}} ({\bf u}_n^\top {\bf x}) {\bf u}_n+ \sum_{\mu(n)=0}
g(\lambda(n)) \tilde  {x}_n {\bf u}_n\\
& \hskip-0.08in  = & \hskip-0.08in  \sum_{\mu(n)\ne 0}
 ({\bf u}_n^\top {\bf x}) {\bf u}_n= \sum_{n=1}^N
 ({\bf u}_n^\top {\bf x}) {\bf u}_n={\bf x},
\end{eqnarray*}
where the first two  equalities follow from  \eqref{filter.Fourier} and \eqref{gaussian.thm.pfeq1}, the third one is true since ${\mathbb E} ({\bf u}_n^\top {\bf x})^2=0$ if $\mu(n)=0$, and the final equality  holds by the  orthonormal basis property of ${\bf u}_1, \ldots, {\bf u}_N$.  This completes the proof of the necessity.    
\end{proof}

By Theorem~\ref{gaussian.thm},  stationarity arises from transmitting white noise through a polynomial channel. Since the composition of two polynomial filters yields another polynomial filter,  we conclude that stationarity is preserved under  transmission through polynomial channels. 

\begin{corollary}
If ${\bf H}$ is a polynomial filter of the graph shift ${\bf S}$ and  $\x$ is  a stationary graph signal, then ${\bf H}{\x}$ is  stationary too.
\end{corollary}

\section {Kalman filtering of stationary signals}
\label{kalman.section}

Let ${\bf x}_k, k\ge 0$, be a family of   random signals on the graph ${\mathcal G}$ 
 governed by the following dynamical system, 
\begin{equation}\label{kalman.def0} {\bf x}_k={\bf A}_{k} {\bf x}_{k-1} +\sigma_{k} {\bf e}_{k}, \  k\ge 1,\end{equation} where 
 the initial signal ${\bf x}_0$  is  assumed to be stationary. Here for each $k\ge 1$,  the transition  matrix ${\bf A}_k$ is a polynomial filter of the graph shift ${\bf S}$, ${\bf e}_k$ is a  standard white noise independent on ${\bf x}_{k-1}$, and 
 $\sigma_k>0$ represents the input noise level.

Write ${\rm cov}({\bf x}_0)=h_0({\mathbf S})$ and ${\bf A}_k=a_k({\mathbf S})$ for some polynomials $h_0$ and $a_k, k\ge 0$. By induction on $k$,  we can show that  stationarity is preserved under the dynamics defined by \eqref{kalman.def0}.
Moreover, the covariance matrices of stationary signals ${\bf x}_k, k\ge 1$,
admit the representation ${\rm cov}({\bf x}_k)
=h_k({\mathbf S})$  for some polynomials $h_k$ satisfying   the recursive relation,
\begin{equation} \label{kalman.eq1} 
h_k= a_{k}^2 h_{k-1} +\sigma_{k}^2, \ \ k\ge 1.\end{equation}
This recurrence captures the evolution of covariance
of stationary signals ${\bf x}_k, k\ge 0$, through 
polynomial channels, with additive noise contributions at each stage.

In this paper, we consider the scenario that  stationary signals ${\bf x}_k, k\ge 0$,  are observed from some polynomial channels and corrupted by  white noises.
 Specifically, the available observations  ${\bf z}_k$ of the true  state signals ${\bf x}_k$  are modeled by
\begin{equation}\label{Kalmanobservation}
 {\bf z}_k={\bf B}_{k} {\bf x}_{k} +\tilde \sigma_k \tilde {\bf e}_k, \  k\ge 1, 
 \end{equation} 
where  the  observation matrices
${\bf B}_k=b_k({\mathbf S})$ are
polynomial filters of the graph shift ${\bf S}$,   
$\tilde {\bf e}_k$ are   standard white noises independent on ${\bf x}_k$, and 
$\tilde \sigma_k>0$  quantify the observation noise levels.

The  Kalman filtering procedure  begins with  an initial  estimator $\widehat {\bf x}_0$ of the true initial ${\bf x}_0$ so that the estimation error $\widehat {\bf x}_0-{\bf x}_0$ is assumed to be stationary. The subsequent Kalman estimators $\widehat {\bf x}_k, k\ge 1$, are computed recursively through two steps:
1) In the  prediction step, we set
\begin{equation}\label{kalman.prediction1}
{\widehat {\mathbf x} }_{k-1/2}={\mathbf A} _{k}
\widehat {\mathbf x} _{k-1}, 
 \end{equation}
 where ${\mathbf A} _{k}$
 is the polynomial filter in \eqref{kalman.def0} to represent the system dynamics;
2) In the  Kalman gain step, we set
 \begin{equation} \label{kalman.prediction2}
{\widehat {\bf x}}_k=
{\widehat {\bf x}}_{k-1/2}+ {\bf K}_k ({\bf z}_k-{\bf B}_{k} {\widehat {\bf x}}_{k-1/2}), 
\end{equation}
where ${\bf B}_k$ is the observation matrix in \eqref{Kalmanobservation} and the  optimal Kalman filter ${\bf K}_k$ is determined by minimizing the mean squared error  $\min_{{\bf K}_k}  {\mathbb E}\| \widehat{{\bf x}}_k-{\bf x}_k\|_2^2$, where $\|{\bf x}\|_2=(\sum_{i\in V} |x(i)|^2)^{1/2}$ denotes the energy of a graph signal ${\bf x}=[x(i)]_{i\in V}$. 
In this paper, we show that this Kalman filtering procedure   preserves stationarity; 
see Theorem \ref{kalman.thm}.


Let
${\bf P}_k= {\rm cov}({\bf x}_{k}-\widehat {\bf x}_{k}),~ k\ge 0$, be the estimation error covariance matrices. 
Following the 
derivation in \cite{ shi2009},  the Kalman filters   ${\bf K}_k$
and the covariance  matrices ${\bf P}_k, k\ge 1$,
 satisfy the following iterative relations:
\begin{eqnarray}\label{Kalman.expression}
{\bf K}_k&\hskip-0.08in  = & \hskip-0.08in  \big({\bf A}_k {\bf P}_{k-1}
{\bf A}_k^\top + \sigma_{k}^2 {\bf I}\big)
{\bf B}_{k}^\top \big({\bf B}_{k} 
{\bf A}_k {\bf P}_{k-1}
{\bf A}_k^\top {\bf B}_k^\top +  \sigma_{k}^2
{\bf B}_{k} {\bf B}_{k}^\top +\tilde \sigma_k^2{\bf I}\big)^{-1}
\end{eqnarray}
and
\begin{equation} \label{Kalman.expression2}
{\bf P}_k  = ({\bf I}-{\bf K}_k {\bf B}_k) \big({\bf A}_k {\bf P}_{k-1}
{\bf A}_k^\top + \sigma_{k}^2 {\bf I}\big), \ k\ge 1.
\end{equation}
Applying \eqref{Kalman.expression}
and \eqref{Kalman.expression2} recursively, we see that the Kalman filters ${\bf K}_k$ and the error covariance matrices  ${\bf P}_k, k\ge 1$, are polynomials of the graph shift ${\bf S}$.

\begin{theorem}\label{kalman.thm} Let ${\bf x}_k, k\ge 0$, be stationary signals governed by the 
dynamical system \eqref{kalman.def0}, and $\widehat {\bf x}_k, k\ge 0$, be the Kalman estimators defined  by \eqref{kalman.prediction1} and \eqref{kalman.prediction2}. Then 
the estimators $\widehat {\bf x}_k$ and the estimation errors
$\widehat {\bf x}_k-{\bf x}_k, k\ge 1$, are stationary.
\end{theorem}

Write ${\bf K}_k= g_k({\mathbf S})$ and ${\bf P}_k=p_k({\mathbf S})$ for some polynomials $g_k$ and $p_k, k\ge 1$. Denote
 the minimal polynomial of the graph shift ${\bf S}$ by $p_{\mathbf S}$, and we say that   $f_1=f_2 \ \ {\rm mod} \  p_{\mathbf S}$ for polynomials $f_1$ and $f_2$  if $(f_1-f_2)/p_{\mathbf S}$ is a polynomial.
Applying \eqref{Kalman.expression}
and \eqref{Kalman.expression2} recursively, one may verify that these polynomials satisfy the following:
 \begin{equation}\label{kalmanfilter.Fourier}
 (a_k^2b_k^2 p_{k-1}+\sigma_k^2 b_k^2+\tilde \sigma_k^2) g_k= 
 (a_k^2 p_{k-1}+\sigma_k^2) b_k
  \ {\rm mod} \ p_{\mathbf S},
\end{equation}
and
\begin{equation} \label{covariance.Fourier}
(a_k^2b_k^2 p_{k-1}+\sigma_k^2 b_k^2+\tilde \sigma_k^2) p_k = \tilde \sigma_k^2 (a_k^2 p_{k-1}+\sigma_k^2) \  {\rm mod} \ p_{\mathbf S}.
\end{equation}
 
By \eqref{covariance.Fourier} and the observation that $p_{\mathbf S}({\mathbf S})=0$, 
we obtain the following expression for the estimation error covariance for the Kalman filtering:
\begin{equation}\label{Pk.eq2}
{\bf P}_k
 = \tilde \sigma_k^2 ({\bf A}_k^2 {\bf P}_{k-1}+\sigma_k^2 {\bf I}) 
 ( {\bf A}_k^2 {\bf B}_k^2 {\bf P}_{k-1}+\sigma_k^2 {\bf B}_k^2+\tilde \sigma_k^2{\bf I})^{-1}.
\end{equation}
Considering the scenario that the observation matrices  ${\bf B}_k$ are all-pass filters (and hence they are invertible), one may neglect  dynamics in the system \eqref{kalman.def0} and  consider the  inverse problem \eqref{Kalmanobservation} only. Then  a static inverse filtering estimator for the state signal  ${\bf x}_k$ is given by 
\begin{equation}\label{inversereconstruction}
   \tilde {\bf x}_k= {\bf B}_k^{-1} {\bf z}_k,\ k\ge 1.
\end{equation}
For the above  estimator $\tilde {\bf x}_k$, one may verify that the
estimation error covariance  matrix is 
${\rm cov}(\tilde {\bf x}_k- {\bf x}_{k})= \tilde \sigma_k^2 ({\bf B}_k)^{-2}
$.

   We say that 
   positive semi-definite matrices ${\bf A}$ and ${\bf B}$ satisfy ${\bf A}\prec {\bf B}$ if ${\mathbf B}-{\bf A}$ is strictly positive definite. 
  Comparing the inverse filtering approach \eqref{inversereconstruction} with the Kalman filtering procedure, we observe that the Kalman filtering  incorporates both the system dynamics and noise structure, yielding a more accurate and adaptive estimation than the 
  inverse filtering. 

\begin{corollary} \label{Kalmaninverse.cor}
Let  the observation matrices ${\bf B}_k, k\ge 1$, are  all-pass filters,
and  $\tilde {\bf x}_k, k\ge 1$, are the reconstructed signals from solving the inverse problem 
 \eqref{inversereconstruction}.
Then 
     \begin{equation}
     {\rm cov}(\widehat {\bf x}_k-{\bf x}_k)\prec 
     {\rm cov}(\tilde {\bf x}_k-{\bf x}_k), \ k\ge 1.
     \end{equation}
\end{corollary}

Consider another scenario in which the observation noise levels
$\tilde \sigma_k, k\ge 1$, are significantly high, which renders the  observations ${\bf z}_k$ unreliable for 
estimating the true state signal  ${\bf x}_k$. In such cases, a naive  strategy is to set the estimator to the zero signal, i.e.,
$\bar {\bf x}_k= {\mathbb E} {\bf x}_k=0$, which satisfy the
dynamic system \eqref{kalman.def0} in the absence of noises,  i.e.,  when $\sigma_k=0, k\ge 1$. 
By \eqref{kalman.eq1} and \eqref{Pk.eq2}, we  conclude that, despite the high observation noise, the Kalman filter adaptively balances the prediction and correction steps, yielding a better estimation than the zero estimator,\begin{equation}
     {\rm cov}(\widehat {\bf x}_k-{\bf x}_k)\prec  {\rm cov}({\bf x}_k-\bar {\bf x}_k)= {\rm cov}({\bf x}_k), \ k\ge 1.
     \end{equation}
This advantage arises from the  Kalman filter’s ability to exploit prior information encoded in the system model and the covariance structure, even when the  observations are weakly informative. 
 As a result, Kalman filtering consistently outperforms the default zero-signal strategy in high-noise environments.

\section{Numerical demonstration}

In this section, we evaluate the performance of the proposed Kalman filtering method on a dynamical system  over the unweighted cycle graph 
${\mathcal C}_N=(V_N, E_N, {\bf W}_N)$, where $V_N=\{1, \ldots, N\}$ and 
$E_N=\{(n, m)~ | ~  m-n=\pm 1 \ {\rm mod}\  N\}$ and ${\bf W}_N=[W_N(n,m)]_{1\le n, m\le N}$ with $W_N(n,m)=1$ for $(n,m)\in E_N$ and $0$ otherwise.  In our simulations, we  adopt the graph  Laplacian  ${\mathbf L}=2{\bf I}-{\bf W}_N$ as the graph shift ${\bf S}$ to define 
stationary graph signals. 
In our simulations, we also assume the dynamical  system and observation models are time-invariant, i.e.,   both the state and observation matrices, as well as noise levels, remain constant over time. Specifically, the dynamical system evolves according to \begin{equation}\label{circlekalman.def0} {\bf x}_k = {\bf A} {\bf x}_{k-1} + \sigma {\bf e}_k,  \end{equation} and the observations are given by \begin{equation}\label{circleKalmanobservation} {\bf z}_k = {\bf B} {\bf x}_k + \tilde{\sigma} \tilde{\bf e}_k, \quad 1\le k \le M, \end{equation} where 
${\bf A}$ and ${\bf B}$
 are polynomials  of the graph Laplacian, $\sigma, \tilde \sigma>0$ are the noise levels, and ${\bf e}_k, \tilde {\bf e}_k$ are i.i.d. Gaussian noises with normal distribution, and $M$  represents the  number of time steps of our simulations. The initial state
 ${\bf x}_0$ and its estimate $\widehat {\bf x}_0$ are both assumed to be zero signals.
 Denote the reconstructed stationary signals  via the Kalman filtering procedure \eqref{kalman.prediction2}
and via the inverse filtering method \eqref{inversereconstruction} by  $\widehat {\bf x}_k$ and $\tilde {\bf x}_k, 1\le k\le M$, respectively.
We remark that  in the procedure
 \eqref{inversereconstruction} to obtain $\tilde {\bf x}_k, 1\le k\le M$, 
the inverse of the observation matrix will be replaced by its Moore-Penrose pseudo-inverse if it is not invertible.

 Fig. \ref{Kalman.fig} illustrates the performance of Kalman filtering and inverse filtering for the dynamical system \eqref{circlekalman.def0} over the cycle graph ${\mathcal C}_N$ with $N=30$  and $M=100$ steps, where  ${\bf A}={\bf L}/4$,
 ${\mathbf B} 
 ={\bf W}_N/2$, $0\le \sigma\le 1$ and $0\le \tilde \sigma\le 1$.
The top  plots of Fig. \ref{Kalman.fig} display 
the average relative reconstruction errors
$\min\big(\frac{1}{2}\log_{10}
\frac{1}{M}\sum_{k=1}^M\|\widehat {\bf x}_k-{\bf x}_k\|_2^2/\|{\bf x}_k\|_2^2, \frac{1}{2}\big)$ and
$\min\big(\frac{1}{2} \log_{10}
\frac{1}{M}\sum_{k=1}^M\|\tilde {\bf x}_k-{\bf x}_k\|_2^2/\|{\bf x}_k\|_2^2, \frac{1}{2}\big)$
for the Kalman filtering  (left) and inverse filtering (right)  over 30 trials, respectively. The  heatmaps  on the top row of Fig.  \ref{Kalman.fig} highlight the clear advantage of Kalman filtering. Unlike the inverse filtering, which relies on a fixed linear inversion and ignores temporal dynamics,
the Kalman filtering continuously refines its estimates by integrating both the system model and noisy observations, resulting in significantly lower errors, especially when the observation noise $\tilde\sigma$ increases.
 
Presented at the bottom of Fig. \ref{Kalman.fig}  is a comparison of reconstructed signals and their  corresponding energy profiles
with the dynamical noise level $\sigma=0.3$ and  observation noise level  $\tilde \sigma=0.5$, respectively. On the bottom left, the energy sequence $\|{\bf x}_k\|_2, 1\le k\le M,$ of the original state signal (red) is shown alongside the reconstruction energies $\|\widehat {\bf x}_k\|_2, 1\le k\le M,$ via Kalman filtering (green), and the  reconstruction energies   $\|\tilde {\bf x}_k\|_2, 1\le k\le M,$ via inverse filtering (blue), respectively.
Presented at the bottom right is the temporal trajectory  at vertex 8: the original signal ${\bf x}_k(8)$ (red), the Kalman-filtered reconstruction signal
 $\widehat {\bf x}_k(8)$ (green), and the inverse-filtered reconstruction signal $\tilde {\bf x}_k(8)$ (blue), over the  time range $1\le k\le M$.
The numerical results indicate that the inverse filtering method amplifies uncertainty by back-projecting measurement noise through the observation matrix 
${\bf B}$, resulting in reconstructions with high variance and instability. In contrast, Kalman filtering leverages the state transition dynamics encoded in ${\bf A}$ to effectively suppress noise and accurately track the underlying signal trend. This dynamical integration enables the Kalman filtering approach to significantly outperform the static inverse filtering method,  offering superior stability and estimation accuracy.

\begin{figure}[t] 

  \centering
  \centerline{\includegraphics[width=16.8cm, height=6.8cm]{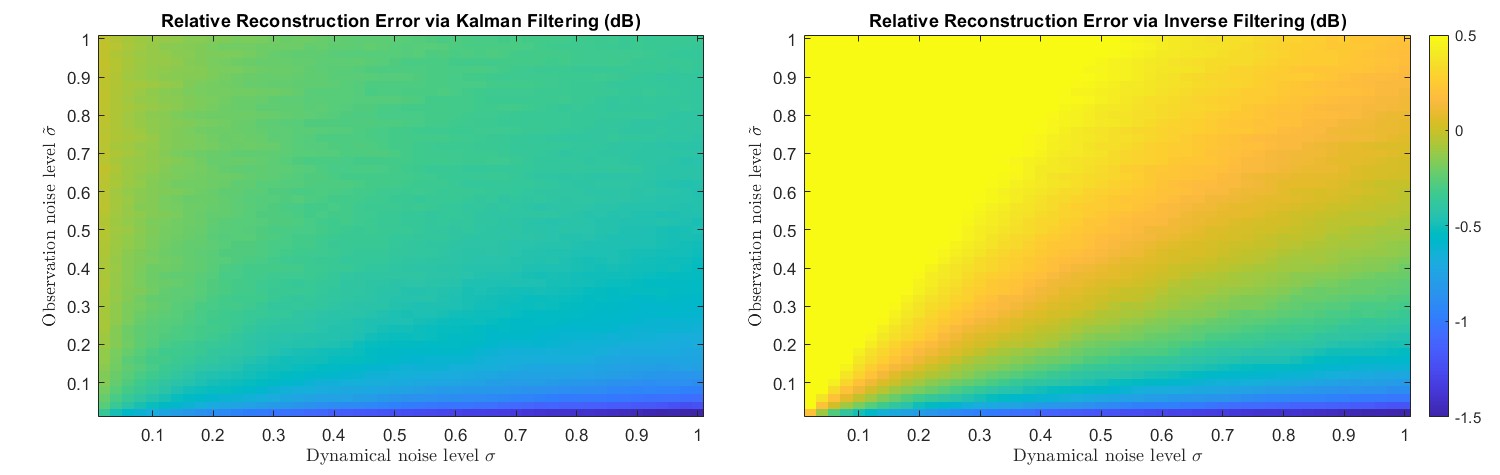}}

  \centering
  \includegraphics[width=8.5cm, height=5.8cm]{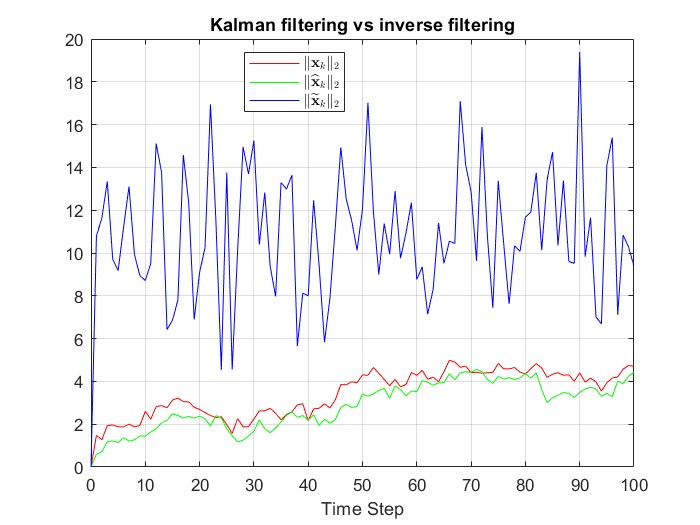}
  \includegraphics[width=8.5cm, height=5.8cm]{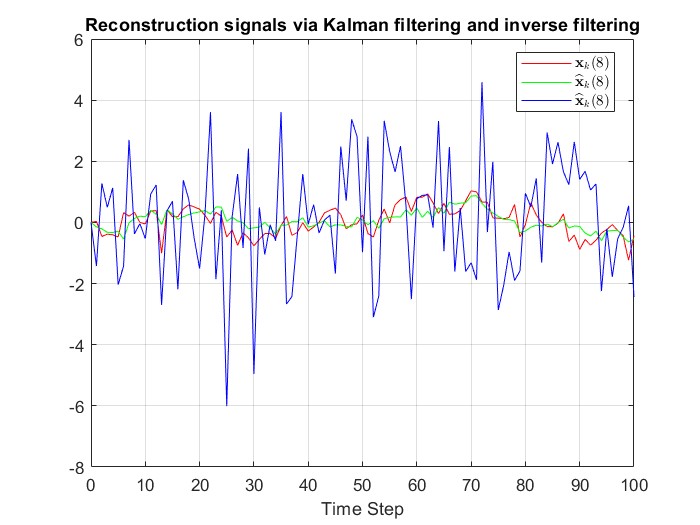}
\caption{Signal reconstruction via Kalman filtering and inverse filtering.}
\label{Kalman.fig}
\end{figure}

\bigskip

{\bf Acknowledgment}:\ 
The authors wish to express their  gratitude to Cheng Cheng  for her valuable 
insights into the Kalman filter.

\end{document}